\documentclass[pra,aps,showpacs,onecolumn,twoside,superscriptaddress]{revtex4}

\usepackage{amsmath,amsfonts,amssymb,caption,color,epsfig,graphics,graphicx,hyperref,latexsym,mathrsfs,revsymb,theorem,url,verbatim,epstopdf}

\hypersetup{colorlinks,linkcolor={blue},citecolor={red},urlcolor={blue}}

\newtheorem{definition}{Definition}
\newtheorem{proposition}[definition]{Proposition}
\newtheorem{lemma}[definition]{Lemma}

\newtheorem{theorem}[definition]{Theorem}
\newtheorem{corollary}[definition]{Corollary}
\newtheorem{conjecture}[definition]{Conjecture}

\newtheorem{remark}[definition]{Remark}
\newtheorem{example}[definition]{Example}
\newtheorem{question}[definition]{Question}

\def\bcj{\begin{conjecture}}
\def\ecj{\end{conjecture}}
\def\bcr{\begin{corollary}}
\def\ecr{\end{corollary}}
\def\bd{\begin{definition}}
\def\ed{\end{definition}}
\def\bea{\begin{eqnarray}}
\def\eea{\end{eqnarray}}
\def\bem{\begin{enumerate}}
\def\eem{\end{enumerate}}
\def\bex{\begin{example}}
\def\eex{\end{example}}
\def\bim{\begin{itemize}}
\def\eim{\end{itemize}}
\def\bl{\begin{lemma}}
\def\el{\end{lemma}}
\def\bma{\begin{bmatrix}}
\def\ema{\end{bmatrix}}
\def\bpf{\begin{proof}}
\def\epf{\end{proof}}
\def\bpp{\begin{proposition}}
\def\epp{\end{proposition}}
\def\bqu{\begin{question}}
\def\equ{\end{question}}
\def\br{\begin{remark}}
\def\er{\end{remark}}
\def\bt{\begin{theorem}}
\def\et{\end{theorem}}

\def\squareforqed{\hbox{\rlap{$\sqcap$}$\sqcup$}}
\def\qed{\ifmmode\squareforqed\else{\unskip\nobreak\hfil
\penalty50\hskip1em\null\nobreak\hfil\squareforqed
\parfillskip=0pt\finalhyphendemerits=0\endgraf}\fi}
\def\endenv{\ifmmode\;\else{\unskip\nobreak\hfil
\penalty50\hskip1em\null\nobreak\hfil\;
\parfillskip=0pt\finalhyphendemerits=0\endgraf}\fi}
\newenvironment{proof}{\noindent \textbf{{Proof.~} }}{\qed}
\def\Dbar{\leavevmode\lower.6ex\hbox to 0pt
{\hskip-.23ex\accent"16\hss}D}
\makeatletter
\def\url@leostyle{%
  \@ifundefined{selectfont}{\def\UrlFont{\sf}}{\def\UrlFont{\small\ttfamily}}}
\makeatother
\def\bcj{\begin{conjecture}}
\def\ecj{\end{conjecture}}
\def\bcr{\begin{corollary}}
\def\ecr{\end{corollary}}
\def\bd{\begin{definition}}
\def\ed{\end{definition}}
\def\bea{\begin{eqnarray}}
\def\eea{\end{eqnarray}}
\def\bem{\begin{enumerate}}
\def\eem{\end{enumerate}}
\def\bex{\begin{example}}
\def\eex{\end{example}}
\def\bim{\begin{itemize}}
\def\eim{\end{itemize}}
\def\bl{\begin{lemma}}
\def\el{\end{lemma}}
\def\bpf{\begin{proof}}
\def\epf{\end{proof}}
\def\bpp{\begin{proposition}}
\def\epp{\end{proposition}}
\def\bqu{\begin{question}}
\def\equ{\end{question}}
\def\br{\begin{remark}}
\def\er{\end{remark}}
\def\bt{\begin{theorem}}
\def\et{\end{theorem}}

\def\btb{\begin{tabular}}
\def\etb{\end{tabular}}

\newcommand{\nc}{\newcommand}

\def\a{\alpha}
\def\b{\beta}

\def\d{\delta}
\def\e{\epsilon}

\def\i{\iota}

\def\l{\lambda}

\def\r{\rho}
\def\s{\sigma}

 \nc{\bbA}{\mathbb{A}} \nc{\bbB}{\mathbb{B}} \nc{\bbC}{\mathbb{C}}
 \nc{\bbD}{\mathbb{D}} \nc{\bbE}{\mathbb{E}} \nc{\bbF}{\mathbb{F}}
 \nc{\bbG}{\mathbb{G}} \nc{\bbH}{\mathbb{H}} \nc{\bbI}{\mathbb{I}}
 \nc{\bbJ}{\mathbb{J}} \nc{\bbK}{\mathbb{K}} \nc{\bbL}{\mathbb{L}}
 \nc{\bbM}{\mathbb{M}} \nc{\bbN}{\mathbb{N}} \nc{\bbO}{\mathbb{O}}
 \nc{\bbP}{\mathbb{P}} \nc{\bbQ}{\mathbb{Q}} \nc{\bbR}{\mathbb{R}}
 \nc{\bbS}{\mathbb{S}} \nc{\bbT}{\mathbb{T}} \nc{\bbU}{\mathbb{U}}
 \nc{\bbV}{\mathbb{V}} \nc{\bbW}{\mathbb{W}} \nc{\bbX}{\mathbb{X}}
 \nc{\bbZ}{\mathbb{Z}}

 \nc{\bA}{{\bf A}} \nc{\bB}{{\bf B}} \nc{\bC}{{\bf C}}
 \nc{\bD}{{\bf D}} \nc{\bE}{{\bf E}} \nc{\bF}{{\bf F}}
 \nc{\bG}{{\bf G}} \nc{\bH}{{\bf H}} \nc{\bI}{{\bf I}}
 \nc{\bJ}{{\bf J}} \nc{\bK}{{\bf K}} \nc{\bL}{{\bf L}}
 \nc{\bM}{{\bf M}} \nc{\bN}{{\bf N}} \nc{\bO}{{\bf O}}
 \nc{\bP}{{\bf P}} \nc{\bQ}{{\bf Q}} \nc{\bR}{{\bf R}}
 \nc{\bS}{{\bf S}} \nc{\bT}{{\bf T}} \nc{\bU}{{\bf U}}
 \nc{\bV}{{\bf V}} \nc{\bW}{{\bf W}} \nc{\bX}{{\bf X}}
 \nc{\bZ}{{\bf Z}}

\nc{\cA}{{\cal A}} \nc{\cB}{{\cal B}} \nc{\cC}{{\cal C}}
\nc{\cD}{{\cal D}} \nc{\cE}{{\cal E}} \nc{\cF}{{\cal F}}
\nc{\cG}{{\cal G}} \nc{\cH}{{\cal H}} \nc{\cI}{{\cal I}}
\nc{\cJ}{{\cal J}} \nc{\cK}{{\cal K}} \nc{\cL}{{\cal L}}
\nc{\cM}{{\cal M}} \nc{\cN}{{\cal N}} \nc{\cO}{{\cal O}}
\nc{\cP}{{\cal P}} \nc{\cQ}{{\cal Q}} \nc{\cR}{{\cal R}}
\nc{\cS}{{\cal S}} \nc{\cT}{{\cal T}} \nc{\cU}{{\cal U}}
\nc{\cV}{{\cal V}} \nc{\cW}{{\cal W}} \nc{\cX}{{\cal X}}
\nc{\cZ}{{\cal Z}}

\nc{\hA}{{\hat{A}}} \nc{\hB}{{\hat{B}}} \nc{\hC}{{\hat{C}}}
\nc{\hD}{{\hat{D}}} \nc{\hE}{{\hat{E}}} \nc{\hF}{{\hat{F}}}
\nc{\hG}{{\hat{G}}} \nc{\hH}{{\hat{H}}} \nc{\hI}{{\hat{I}}}
\nc{\hJ}{{\hat{J}}} \nc{\hK}{{\hat{K}}} \nc{\hL}{{\hat{L}}}
\nc{\hM}{{\hat{M}}} \nc{\hN}{{\hat{N}}} \nc{\hO}{{\hat{O}}}
\nc{\hP}{{\hat{P}}} \nc{\hR}{{\hat{R}}} \nc{\hS}{{\hat{S}}}
\nc{\hT}{{\hat{T}}} \nc{\hU}{{\hat{U}}} \nc{\hV}{{\hat{V}}}
\nc{\hW}{{\hat{W}}} \nc{\hX}{{\hat{X}}} \nc{\hZ}{{\hat{Z}}}

\nc{\hn}{{\hat{n}}}

\def\max{\mathop{\rm max}}
\def\min{\mathop{\rm min}}

\def\tr{\mathop{\rm Tr}}

\def\SU{{\mbox{\rm SU}}}

\newcommand{\bra}[1]{\langle#1|}
\newcommand{\ket}[1]{|#1\rangle}
\newcommand{\proj}[1]{| #1\rangle\!\langle #1 |}
\newcommand{\ketbra}[2]{|#1\rangle\!\langle#2|}

\def\Dbar{\leavevmode\lower.6ex\hbox to 0pt
{\hskip-.23ex\accent"16\hss}D}

\begin{document}

\title{Detection of genuine tripartite entanglement by two bipartite entangled states}


\date{\today}

\pacs{03.65.Ud, 03.67.Mn}

\author{Yize Sun}\email[]{sunyize@buaa.edu.cn}
\affiliation{School of Mathematical Sciences, Beihang University, Beijing 100191, China}

\author{Lin Chen}\email[]{linchen@buaa.edu.cn (corresponding author)}
\affiliation{School of Mathematical Sciences, Beihang University, Beijing 100191, China}
\affiliation{International Research Institute for Multidisciplinary Science, Beihang University, Beijing 100191, China}

\begin{abstract}
It is an interesting problem to construct genuine tripartite entangled states based on the collective use of two bipartite entangled  states. We  consider the case that the states are two-qubit Werner states, 
we construct the interval of parameter of Werner states such that the tripartite state is genuine entangled. Further, we present the way of detecting the tripartite genuine entanglement using current techniques in experiments. We also investigate the lower bound of genuine multipartite entanglement concurrence.

\end{abstract}

\maketitle


\section{Introduction}
Genuine multipartite entanglement offers significant advantages compared with bipartite entanglement in quantum tasks \cite{Horodecki2007Quantum,Nielsen2011Quantum, Divincenzo1995Quantum}. 
 Multipartite private states from which secret keys are directly accessible to trusted partners are genuinely multipartite entangled states \cite{das2019universal}.
Furthermore, genuine entanglement (GE) is the basic ingredient in measurement-based quantum computation \cite{Briegel2009Measurement,RaussendorfA}, and is beneficial in various quantum communication protocols \cite{DeQuantum}. 
Many efforts have been devoted towards the detection of genuine entanglement. For instance, GE can be computed efficient by the generalized geometric measure \cite{roy2019computable}, a series of linear and nonlinear entanglement witnesses  \cite{sun2019improved,HuberDetection,HuberWitnessing,DeMultipartite,Huber2013Entropy,SperlingMultipartite,Augusiak_2009,JungnitschTaming,CoffmanDistributed}, generalized concurrence  \cite{MaMeasure,ChenImproved,Hong2012Measure,Gao2014On}  and Bell-like inequalities \cite{BancalDevice}.    
Although these methods were derived and a characterization in terms of semidefinite programs was developed \cite{Lancien_2015}, the problem of detecting GE remains far from being satisfactorily solved.


\begin{figure}[t]
	\includegraphics[scale=0.4,angle=0]{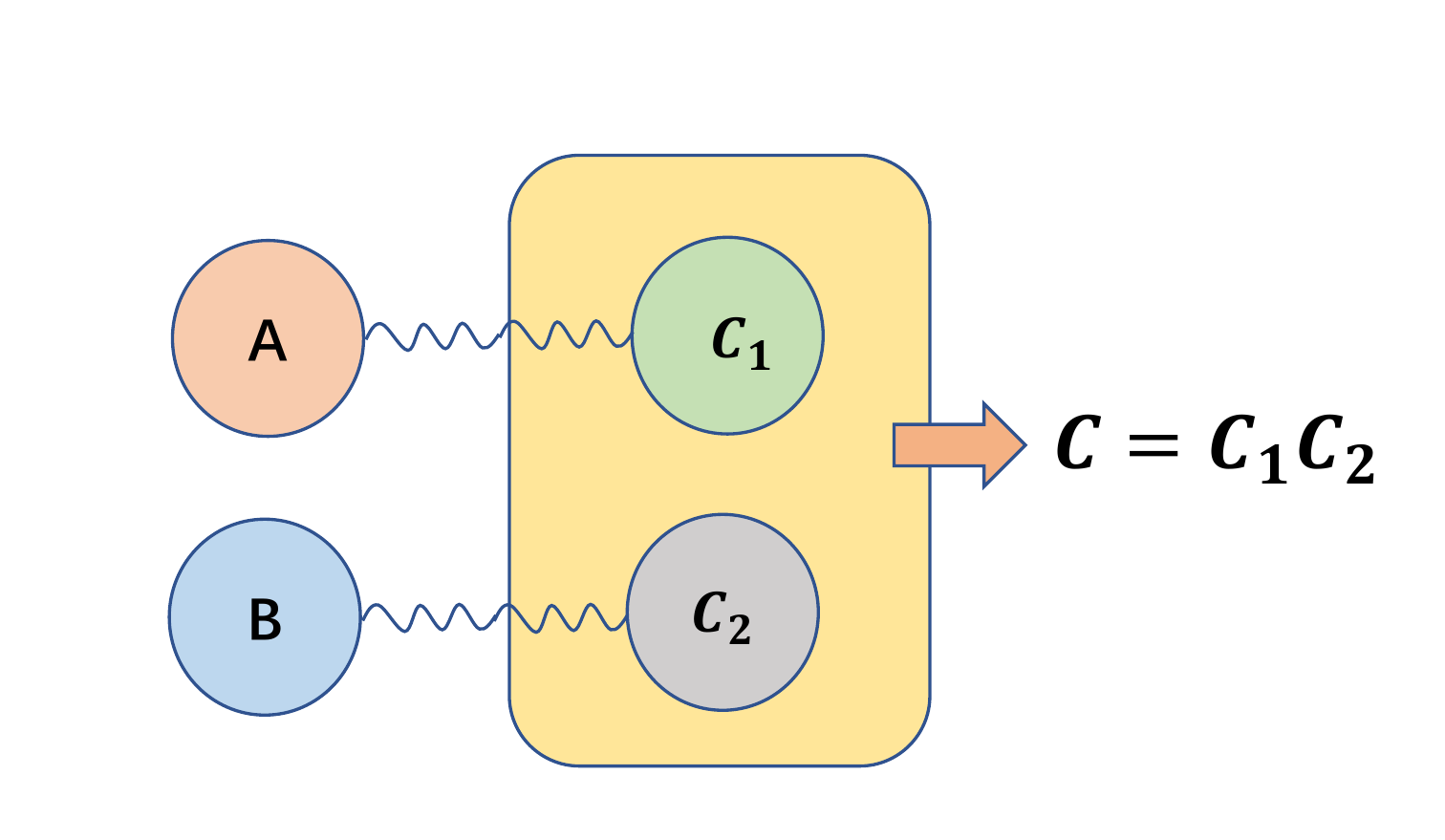}
	\caption{Two bipartite states of systems $A, C_1$ and $B, C_2$ can construct a tripartite state $\r_{ABC}$ of systems $A, B$ and $C$, where $C=C_1C_2$. If $\a_{AC_1}$ and $\b_{BC_2}$ are entangled, then we conjecture that $\r_{ABC}$ is a tripartite genuine entangled state. We partially prove the conjecture in Theorem \ref{eq:theorem5} by showing that $\r_{ABC}$ is a tripartite gunuine entangled state for the interval [-1,-0.94066] of parameter of Werner states.
	}
\label{fig:ptri}
\end{figure}


In the one-dimensional cluster-Ising model \cite{Giampaolo_2014}, the genuine tripartite entanglement between three adjacent spins is the only source of genuine multipartite entanglement.
The tripartite state is either a GE state or a biseparable state. The determination of bipartite entanglement has much more tools in theory and experiments  \cite{VanMultipartite,WuQuantum} than that of tripartite states \cite{HuberDetection}. It is known that all positive partial transpose (PPT) state  are never distillable, so distillable state must be non-positive partial transpose (NPT) state \footnote{If a quantum state is still positive after partial transpose, then the state is PPT. Otherwise, it is NPT.}. The distillability problem is a main open problem of quantum information. It asks whether bipartite NPT states can be asymptotically converted into pure entangled states under local operations and classical communications (LOCC)  \cite{DivincenzoEvidence,D1999Distillability}. To solve the problem,  there are several attempts by converting states into Werner states by LOCC \cite{KrausCharacterization,Bandyopadhyay,ViannaDistillability}. On the other hand, progress towards distilling entangled states of given dimensions or deficient rank has been made steadily. For example, in \cite{ChenDistillability}, it shows that a bipartite NPT quantum state of rank four is distillable.   
 Besides, any biseparable state is a PPT mixture \cite{JungnitschTaming}.
 As far as we know, the relation between genuine tripartite entangled state and NPT or PPT states is not well studied.

In this paper,  we investigate the tripartite GE of tensor product of two NPT Werner states  $\r_w(p_1,d)$, $\r_w(p_2,d)$. In \cite{DivincenzoEvidence}, if Werner states form $\r_w(p,d)$ were distillable then it equals to NPT states would be distillable through the reductions.
Then we present the region of detecting GE of constructing tripartite state $\r_{ABC}$ containing $\r_w(p_1,2)_{AC_1}\otimes\r_w(p_2,2)_{BC_2}$ of systems $A$, $B$ and $(C_1C_2)$ in Theorem \ref{eq:theorem5}. Besides, we present the region of parameter $p$ for detecting GE of tensor product of two Werner states in Theorem \ref{eq:theorem5}. There exist a neighborhood $h=[0,0.05934)$ such that $\r_w(-1+\e,2)_{AC_1}\otimes\r_w(-1+\e,2)_{BC_2}$ is a tripartite genuine entangled state for all $\e\in[0,h]$. Then we discuss the realization of Theorem \ref{eq:theorem5} in experienment, and investigate Conjecture \ref{conj:1} when one of $\r_w(p_1,2)_{AC_1}$ and $\r_w(p_2,2)_{BC_2}$ is a PPT entangled state.
We also correct the lower bound for genuine multipartite entanglement concurrence in Theorem \ref{thm:theorem3} and apply the method in Example \ref{eq:example2} and Example \ref{eq:example3}. Then we use it to detect the GE of $\r_w(p_1,2)_{AC_1}\otimes\r_w(p_2,2)_{BC_2}$ in Appendix \ref{eq:appendix2}. 
If the conjecture is true, then we can construct genuine tripartite entanglement state from two Werner states. Furthermore, it is a special case of constructing an $(n+2)$-patite genuine entanglement state from two $(n+1)$-partite states.

In quantum information theory, the strong subadditivity plays a crucial role in nearly every nontrivial insight \cite{LiebProof,HaydenStructure,Nielsen2007Quantum}. For a tripartite state $\r_{ABC}$, the inequality says
$
S(\r_{AC})+S(\r_{BC})\geq S(\r_{ABC})+S(\r_C),
$
where $S(\a)$ is the von Neumann entropy of quantum state $\a$. The inequality is saturated when there exists a decomposition of system $C$ as
$\mathcal{H}_C=\mathcal{H}_{C_1}\otimes \mathcal{H}_{C_2}$ into a tensor product $\r_{ABC}=\r_{AC_1}\otimes\r_{BC_2}$, where $\r_{AC_1}\in\mathcal{B}(\mathcal{H}_A\otimes\mathcal{H}_{C_1})$, $\r_{BC_2}\in\mathcal{B}(\mathcal{H}_B\otimes\mathcal{H}_{C_2})$. Indeed it is a special case of the necessary and sufficient condition constructed in \cite{HaydenStructure}. Studying the conjecture in Fig. \ref{fig:ptri} thus may help understand the relation between the genuine entanglement and strong subaddivity.

The rest of this paper is organized as follows. In Sec. \ref{sec:pre} we introduce the preliminary knowledge used in this paper. In Sec. \ref{sec:lower} we present our results on the detection of genuine entanglement. In Sec. \ref{sec:4} we investigate the lower bound of GE concurrence. This paper ends up with a conculsion  in Sec. \ref{sec:con}.

\section{Preliminaries}
\label{sec:pre}
In this section we introduce the preliminary knowledge used in this paper. First, genuinely multiparty entangled state is a particularly useful notion in the theory of entanglement but also have found an application, for example, in quantum error correction and cryptography. Besides, genuinely multiparty entangled states can construct genuinely entangled subspace  \cite{demianowicz2019approach}. Next, we define the notations of generators of special unitary group $\SU(d)$ and the $k$ norm for $m\times n$ matrix. Then we present two methods to detect GE of quantum state in Theorem \ref{eq:theorem2} and Theorem \ref{eq:theorem1}.


We review the definition of genuinely multiparty entangled state. A multipartite quantum state that is not separable with respect to any bipartition is said to be genuine multipartite entangled \cite{G2009Entanglement}. Denote $\mathcal{H}_i^d$ as $d$-dimensional Hilbert spaces, $i=1,2,\cdots,n$. An $n$-partite pure state $\ket{\psi}\in\mathcal{B}(\mathcal{H}_1\otimes\mathcal{H}_2\otimes\cdots\otimes\mathcal{H}_n)$ is called biseparable if it can be written as
\begin{eqnarray}
\ket{\psi}=\ket{\psi_{S_1}}\otimes\ket{\psi_{S_2}},
\end{eqnarray}
where $\ket{\psi_{S_1}}\in\mathcal{B}(\mathcal{H}_{S_1})=\mathcal{B}(\mathcal{H}_{s_1}\otimes\mathcal{H}_{s_2}\otimes\cdots\otimes\mathcal{H}_{s_k})$, $\ket{\psi_{S_2}}\in\mathcal{B}(\mathcal{H}_{S_2})=\mathcal{B}(\mathcal{H}_{s_{k+1}}\otimes\mathcal{H}_{s_{k+2}}\otimes\cdots\otimes\mathcal{H}_{s_n})$, and $\{s_1, \cdots, s_k|s_{k+1}, \cdots, s_{n}\}$ is a particular order of $\{1,2,\cdots,n\}$. An $n$-partite mixed state $\r$ is
biseparable if it can be written as a convex combination
of biseparable pure states
\begin{eqnarray}
\r=\sum_ip_i\ketbra{\psi_i}{\psi_i},
\end{eqnarray}
where $0<p_i\leq1$, $\sum p_i=1$, and $\ket{\psi_i}$ is biseparable with respect to different bipartitions. Otherwise, it is called genuinely $n$-partite entangled.	


Next, let $\l_i,i=1,\cdots,d^2-1$, denote the generator of the special unitary group $\SU(d)$ \cite{Kimura2003The}. For example, when $d=3$, the generators $\l$-matrices of the unimodular unitary group $\SU(3)$ are as follows,
\begin{eqnarray}
\l_1=
\bma
0&1&0\\
1&0&0\\
0&0&0
\ema,
\quad
\l_2=
\bma
0&-i&0\\
i&0&0\\
0&0&0
\ema,
\quad
\l_3=
\bma
1&0&0\\
0&-1&0\\
0&0&0
\ema,\nonumber
\end{eqnarray}
\begin{eqnarray}
\l_4=
\bma
0&0&1\\
0&0&0\\
1&0&0
\ema,
\quad
\l_5=
\bma
0&0&-i\\
0&0&0\\
i&0&0
\ema,
\quad
\l_6=
\bma
0&0&0\\
0&0&1\\
0&1&0
\ema,\nonumber
\end{eqnarray}
\begin{eqnarray}
\l_7=
\bma
0&0&0\\
0&0&-i\\
0&i&0
\ema,
\quad
\l_8=\frac{1}{\sqrt{3}}
\bma
1&0&0\\
0&1&0\\
0&0&-2
\ema,
\end{eqnarray}
and $\tr(\l_k\l_l)=2\d_{kl}$, $k,l=1,2,\cdots,8$. Besides, the generators of $\SU(d)$ are the elements of Bloch vector, which gives the desirable description of the states for $N$-level systems.
Let $I$ be the $d\times d$ identity matrix. Any $\r\in\mathcal{B}( \mathcal{H}_1^d\otimes \mathcal{H}_2^d\otimes \mathcal{H}_3^d)$ can be represented as follows,
\begin{eqnarray}
\label{eq:defr}
\r
&=&
\frac{1}{d^3}I\otimes I\otimes I+\frac{1}{2d^2}(\sum t_i^1\l_i\otimes I\otimes I\nonumber\\
&+&
\sum t_j^2I\otimes \l_j\otimes I+\sum t_k^3I\otimes I\otimes \l_k)\nonumber
\\
&+&
\frac{1}{4d}(\sum t_{ij}^{12}\l_i\otimes \l_j\otimes I+\sum t_{ik}^{13}\l_i\otimes I\otimes \l_k\nonumber
\\
&+&
\sum t_{jk}^{23}I\otimes\l_j\otimes\l_k)+\frac{1}{8}\sum t_{ijk}^{123}\l_i\otimes\l_j\otimes\l_k.
\end{eqnarray}

We know that $\tr(\l_k\l_l)=2\d_{kl}$ from the second section of \cite{Simon2009The}. Then according to (\ref{eq:defr}), we obtain that $t_i^1=\tr(\r\l_i\otimes I\otimes I)$. Similarly, $t_j^2=\tr(\r I\otimes\l_j\otimes I)$, $t_k^3=\tr(\r I\otimes I\otimes\l_k)$, $t_{ij}^{12}=\tr(\r\l_i\otimes \l_j\otimes I)$, $t_{ik}^{13}=\tr(\r\l_i\otimes I\otimes \l_k)$, $t_{jk}^{23}=\tr(\r I\otimes \l_j\otimes \l_k)$, and $t_{ijk}^{123}=\tr(\r\l_i\otimes \l_j\otimes \l_k)$. Set $T^{(1)},T^{(2)},T^{(3)},T^{(12)},T^{(13)},T^{(23)}$, and $T^{(123)}$ to be the vectors with the entries $t_i^1, t_j^2, t_k^3,t_{ij}^{12},t_{ik}^{13},t_{jk}^{23}$ and $t_{ijk}^{123}$, $i,j,k=1,2,\cdots,d^2-1$.
Let $T_{\underline{1}23}$, $T_{\underline{2}13}$ and $T_{\underline{3}12}$ be the matrices with entries $t_{i, (d^2-1)(j-1)+k}=t_{ijk}$, $t_{j, (d^2-1)(i-1)+k}=t_{ijk}$ and $t_{k, (d^2-1)(i-1)+j}=t_{ijk}$, respectively. Then we show three norms for an $m\times n$ matrix $M$ as follows,
\begin{eqnarray}
\|M\|
:=
\sqrt{\sum_{i,j}M_{i,j}^2}=\sqrt{\sum_{i}\s_{i}^2},
\end{eqnarray}
\begin{eqnarray}
\|M\|_{tr}
:=
\tr\sqrt{M^TM}=\sum_i\s_i,
\end{eqnarray}
\begin{eqnarray}
\|M\|_k:=\sum_{i=1}^k\s_i,
\end{eqnarray}
where $\{\s_i\}$ ($i=1,\cdots,\min(m,n)$) denote the singular values of the matrix, which are in non-increasing order. Notice that the last Ky Fan norm is the trace norm. 

Let $\|M\|_k=\sum_{i=1}^k\s_i$ denote the $k$ norm for an $n\times n$ matrix $M$, where $\s_i$, $i=1,\cdots,n$, are the singular values of $M$ in decreasing order.
After by presenting the definition of $M_k(\r)$, then we obtain Theorem \ref{eq:theorem2} in the following.
\begin{theorem}
\label{eq:theorem2}
Consider the average matricization norm
\begin{eqnarray}
M_k(\r)=\frac{1}{3}(\|T_{\underline{1}23}\|_k+\|T_{\underline{2}13}\|_k+\|T_{\underline{3}12\|_k})
\end{eqnarray}
for a tripartite qudit state $\r$. If it holds that
\begin{eqnarray}
\label{eq:ineq2}
M_k(\r)>\frac{2\sqrt{2}}{3}(2\sqrt{k}+1)\frac{d-1}{d}\sqrt{\frac{d+1}{d}}
\end{eqnarray}
for any $k=1,2,\cdots,d^2-1$, then $\r$ is a genuine multipartite entangled.
\end{theorem}

In the following, by using Theorem \ref{eq:theorem2}, we shall detect the genuine entanglement of another type of tripartite states. This state is tensor product of two Werner states, which are invariant under all unitaries of the form $U\otimes U$ \cite{WernerQuantum}. Theorem \ref{eq:theorem2} can  detect GE not only for tripartite qubit systems but for any tripartite qudit system.

Next, denote $\|\cdot\|$ the Frobenius norm of a vector or a matrix $\|M\|=\|M\|_{KF}$ is just the Ky-Fan norm.	From Vicente criterion in \cite{Vicente2011Multipartite}, then we have Theorem \ref{eq:theorem1}.

\begin{theorem}
	\label{eq:theorem1}
 For an arbitrary tripartite qudit state it holds that
	\begin{eqnarray}
	\|T_{jlm}\|>\sqrt{\frac{8(d-1)(d^2-1)}{d^3}},
	\end{eqnarray}
	then  the state is genuine multipartite entangled.
\end{theorem}

The result shows that a high value of this measure can imply not only some entanglement but even genuine multipartite entanglement. The power of this condition increases with the subsystem
dimension improving remarkably on \cite{HuberDetection}.
By applying the condition of Theorem \ref{eq:theorem1}, all the quantities are invariant under local
unitary transformations on the density matrix, then we can detect GE of the states in Example \ref{eq:example3}.
Hence, if the lower bound is already large enough to violate the inequality
given by Theorem \ref{eq:theorem1}, then we conclude with certainty the presence of genuine multipartite entanglement 
. On the analogy of Theorem \ref{eq:theorem1}, we also hope to improve it for larger $d$.

\section{Genuine entanglement detection of  tensor product of two bipartite entangled states}
\label{sec:lower}
In this section, we consider how to construct a tripartite GE state for two bipartite entangled states $\a_{AC_1}$ and $\b_{BC_2}$ by involving the tensor product and the Kronecker product, it provides a systematical method to construct GE states. We mainly investigate the following conjecture \ref{conj:1} for bipartite entangled states.  It's known that each NPT bipartite state can be convert into an NPT Werner state by using LOCC, then we introduce the definition of Werner state at first. The Werner state on $\mathbb{C}^d\otimes \mathbb{C}^d$ is defined as
\begin{eqnarray}
\label{eq:rwpd}
\r_w(p,d)
&:=&
\frac{1}{d^2+pd}(I_d\otimes I_d+p\sum_{i,j=0}^{d-1}\ketbra{i,j}{j,i}),
\end{eqnarray}
where the parameter $p\in[-1,1]$. It has been proved \cite{DivincenzoEvidence} that $\r_w(p,d)$ is (i) separable when $p\in[-\frac{1}{d},1]$;
(ii) NPT and one-copy undistillable when $p\in[-\frac{1}{2},-\frac{1}{d})$; and (iii) NPT and one-copy distillable when $p\in[-1,-\frac{1}{2})$. Hence studying the Werner state with $p=-1/2$ and $1$ would characterize the behavior of Werner states over the whole interval of $p$. This is one of the motivations why we propose the following conjecture.
\begin{conjecture}
\label{conj:1}
(i) Suppose $\a_{AC_1}$  and $\b_{BC_2}$ are two bipartite entangled states of systems $A, C_1$ and systems $B, C_2$, respectively. Then $\a_{AC_1}\otimes\b_{BC_2}$ is a GE state of systems $A, B$ and $C$, where $C=(C_1C_2)$.

(ii) Suppose $\a_{AC_1}$ and $\b_{BC_2}$ are two-qubit entangled states. Then we have $\a_{AC_1}\otimes\b_{BC_2}$ is a tripartite state if and only if there is a neighborhood $[0,h)$, for all $\e\in[0,h)$, the state $\r_w(2,-1+\e)_{AC_1}\otimes\r_w(2,-1+\e)_{AC_1}$ is a tripartite GE state.
\end{conjecture}
One may show that Conjecture \ref{conj:1} (ii) is a special case of (i). We first consider to attack the generic one in Theorem \ref{eq:theorem range}, by introducing a known result from \cite{ShenConstruction}.
\begin{theorem}
	\label{eq:theorem range}
	Conjecture \ref{conj:1} (i) holds if the range of $\a$ or $\b$ is not spanned by product states.
\end{theorem}
Next, to detect the GE of tripartite state $\r_w(2,-1+\e)_{AC_1}\otimes\r_w(2,-1+\e)_{BC_2}$ in Conjecture \ref{conj:1}, we construct a special state $\r$ in Theorem \ref{eq:theorem5}.

\begin{theorem}
\label{eq:theorem5}
Consider a mixed state $\r=\frac{1-x}{64}I+x(\a_{AC_1}\otimes\b_{BC_2})$, where $\a_{AC_1}$ and $\b_{BC_2}$ are two-qubit NPT Werner state, then $\a_{AC_1}\otimes\b_{BC_2}\in \mathcal{H}_{ABC}$ is a tripartite state with $C=C_1C_2$.

(i) When $p_1=p_2=-1$, the region of the GE detection of state $\r$ is maximum for $0.902646<x\leq1$.

(ii) When $x=1$, then the GE of state $\r$ can be detected for $-1\leq p_1\leq-0.940198$, $-1\leq p_2\leq-0.94066$.

\end{theorem}

\begin{proof}
	(i) Suppose $\a_{AC_1}=\r_w(p_1,2)=\frac{1}{(4+2p_1)}(I_2\otimes I_2+p_1\sum_{i,j=0}^1\ketbra{i,j}{j,i})$ and $\b_{BC_2}=\r_w(p_2,2)=\frac{1}{(4+2p_2)}(I_2\otimes I_2+p_2\sum_{i,j=0}^1\ketbra{i,j}{j,i})$, where $p_1,p_2\in[-1,-\frac{1}{2})$.
	It shows that $\r\in\mathcal{B}(\mathcal{H}^2\otimes \mathcal{H}^2\otimes \mathcal{H}^4)$, then we have $\r\in\mathcal{B}(\mathcal{H}^4\otimes \mathcal{H}^4\otimes \mathcal{H}^4)$. The generators $\l$-matrices of $\SU(4)$ are as follows,
	\begin{widetext}
		\begin{eqnarray}
		\l_1=\bma 0&1&0&0\\1&0&0&0\\0&0&0&0\\0&0&0&0\ema,
		\l_2=\bma 0&-i&0&0\\i&0&0&0\\0&0&0&0\\0&0&0&0\ema,
		\l_3=\bma 1&0&0&0\\0&-1&0&0\\0&0&0&0\\0&0&0&0\ema,
		\l_4=\frac{1}{\sqrt{3}}
		\bma 1&0&0&0\\0&1&0&0\\0&0&-2&0\\0&0&0&0\ema,
		\end{eqnarray}
		
		\begin{eqnarray}
		\l_{5}=\frac{1}{\sqrt{6}}\bma 1&0&0&0\\0&1&0&0\\0&0&1&0\\0&0&0&-3\ema,
		\l_6=\bma 0&0&0&0\\0&0&1&0\\0&1&0&0\\0&0&0&0
		\ema,
		\l_7=\bma 0&0&1&0\\0&0&0&0\\1&0&0&0\\0&0&0&0\ema,
		\l_8=\bma 0&0&-i&0\\0&0&0&0\\i&0&0&0\\0&0&0&0\ema,
		\end{eqnarray}
		\begin{eqnarray}
		\l_9=\bma 0&0&0&-i\\0&0&0&0\\0&0&0&0\\i&0&0&0\ema,
		\l_{10}=\bma 0&0&0&0\\0&0&0&-i\\0&0&0&0\\0&i&0&0\ema,
		\l_{11}=\bma 0&0&0&0\\0&0&0&0\\0&0&0&-i\\0&0&i&0\ema,
		\l_{12}=\bma 0&0&0&0\\0&0&-i&0\\0&i&0&0\\0&0&0&0\ema,
		\end{eqnarray}
		\begin{eqnarray}
		\l_{13}=\bma 0&0&0&1\\0&0&0&0\\0&0&0&0\\1&0&0&0\ema,
		\l_{14}=\bma 0&0&0&0\\0&0&0&1\\0&0&0&0\\0&1&0&0\ema,
		\l_{15}=\bma 0&0&0&0\\0&0&0&0\\0&0&0&1\\0&0&1&0
		\ema.
		\end{eqnarray}
	\end{widetext}
	
In the following, we detect GE of $\r$ by Theorem \ref{eq:theorem2}. First, we assume that $P=\ketbra{0}{0}+\ketbra{1}{1}$, where $\ket{0}, \ket{1}\in \mathbb{C}^4$. Then we have $(P\otimes P\otimes I_4)\r(P\otimes P\otimes I_4)=\r$. From the definition of $t_{ijk}^{123}$, we obtain that
\begin{eqnarray}
\label{eq:tijk}
	t_{ijk}^{123}
	&=&
	\tr(\r(\l_i\otimes\l_j\otimes\l_k))\nonumber
	\\
	&=&
	\tr(\r(P\l_iP\otimes P\l_jP\otimes  \l_k)).
\end{eqnarray}
According to (\ref{eq:tijk}) and the form of generators $\l$-matrices, 
we have $t_{ijk}^{123}$ is nonzero when $i,j\in\{1,2,3,4,5\}$. By respectively computing $\|T_{\underline{1}23}\|_k$, $\|T_{\underline{2}13}\|_k$ and $\|T_{\underline{3}12}\|_k$, $k=4$, we obtain that
	\begin{eqnarray}
	\max M_k(\r)
	&=&
	\max\frac{1}{3}(\|T_{\underline{1}23}\|_k+\|T_{\underline{2}13}\|_k+\|T_{\underline{3}12}\|_k)\nonumber
	\\
	&=&
4.37918x
	\end{eqnarray}
	when $p_1=p_2=-1$ in Fig. \ref{fig:p9}. Suppose that $k=4$, then we have
	\begin{eqnarray}
	\max M_k(\r)-\frac{5}{2}\sqrt{\frac{5}{2}}
	&=&
	4.37918x-3.95285
	>0.
	\end{eqnarray}
	By analysing the maximum of the coefficient of $x$ in Fig.
	\ref{fig:p9}, then Theorem \ref{eq:theorem2} can detect the GE for $0.902646 <
	x \leq 1$.
	
	(ii) Consider the boundary condition of $p_1$ and $p_2$ in the region of $x$ that GE can be detected. Then in Fig. \ref{fig:p11}, we know that the boundary condition of $p_1$, $p_2$ is
	\begin{eqnarray}
	\frac{1}{3}(\|T_{\underline{1}23}\|_k+\|T_{\underline{2}13}\|_k+\|T_{\underline{3}12}\|_k)=\frac{5}{2}\sqrt{\frac{5}{2}},
	\end{eqnarray}
	where $x=1$.
	
\begin{figure}[t]
		\includegraphics[scale=0.5,angle=0]{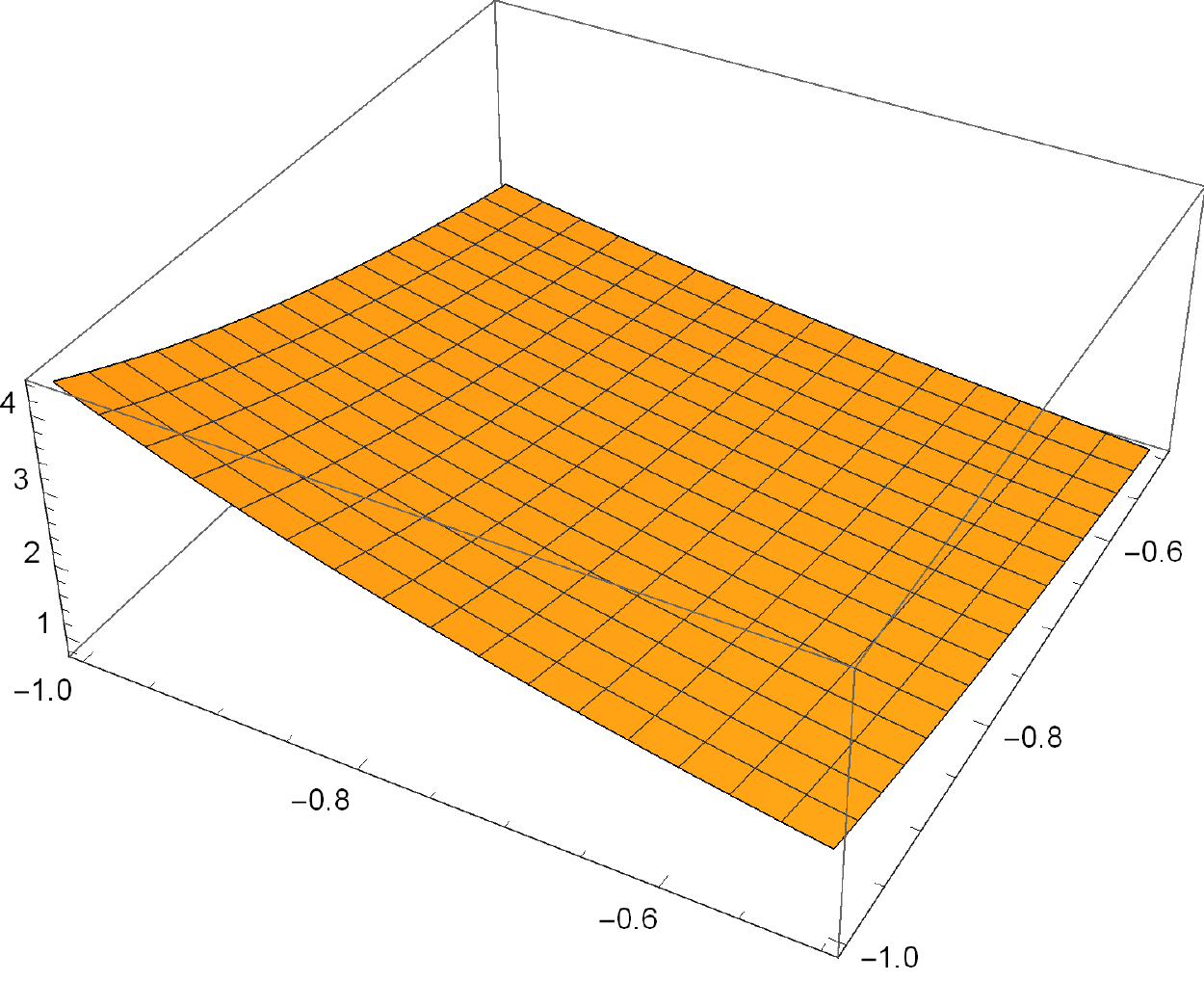}
		\caption{One finds from the figure that when $p_1\rightarrow-1, p_2\rightarrow-1$, the region of $x$ that GE of state $\r$ can be detected is increasing. Specially,  when $p_1=p_2=-1$, the $\max M_k(\r)=4.37918x$ for $0\leq x\leq1$.}
		\label{fig:p9}
	\end{figure}
	
	\begin{figure}[t]
		\includegraphics[scale=0.5,angle=0]{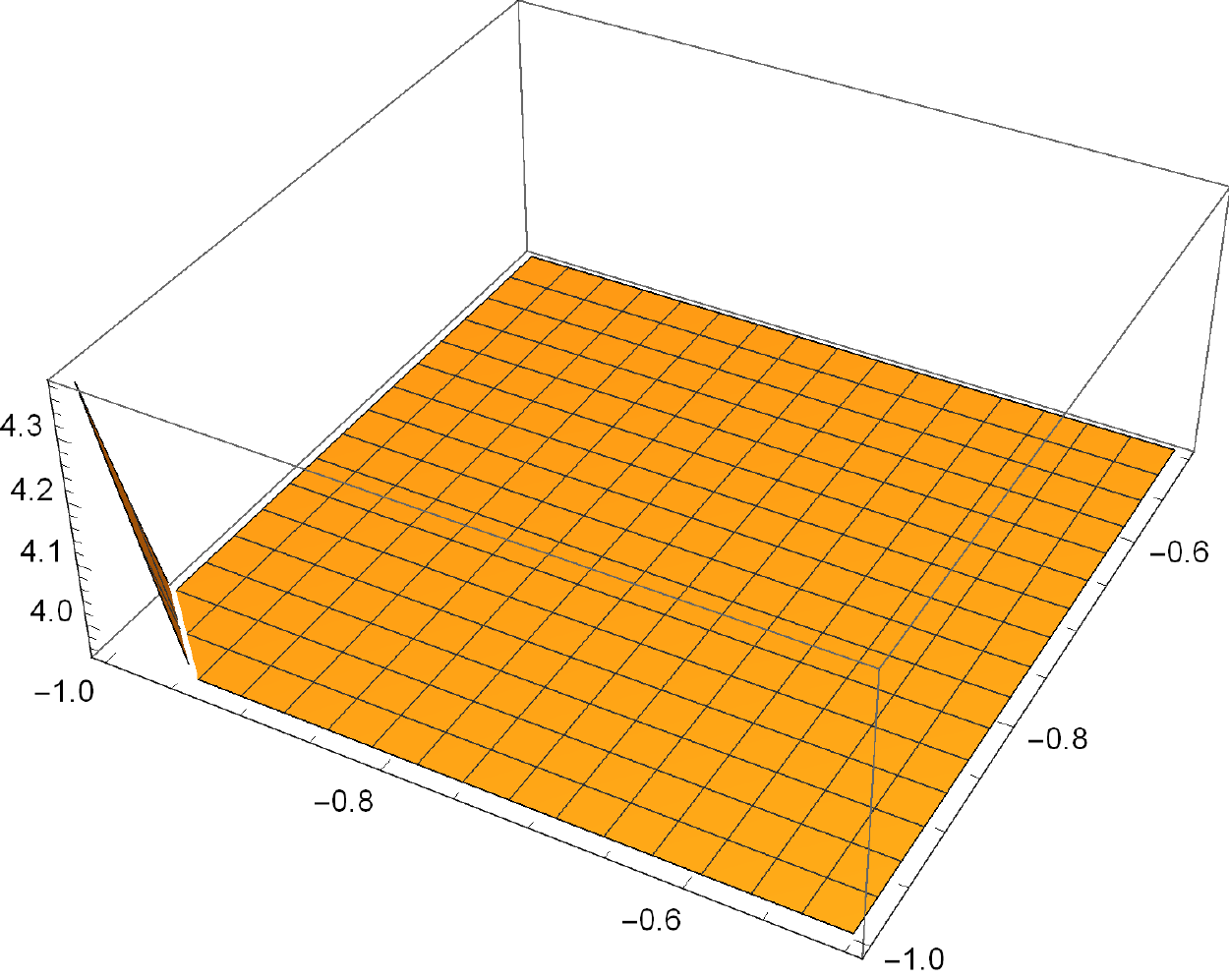}
		\caption{One finds from the figure that $-1\leq p_1\leq -0.940198$ and $-1\leq p_2\leq -0.94066$ in the region of $x$ that GE can be detected.}
		\label{fig:p11}
	\end{figure}

\end{proof}

To compare the region of $x$ that GE of state $\r$ can be detected, then we introduce  another method  in the next section. In  Appendix \ref{eq:appendix2}, by using other two different ways to detect GE of state $\r$.
Thus, we obtain that the maximum region of $x$ that can detect $\r$ by using Theorem \ref{eq:theorem2}. Let $Q=\ketbra{1}{1}+\ketbra{2}{2}$, then we have
\begin{eqnarray}
\label{eq:distillable}
(Q\otimes Q)\r_w(p,d)(Q\otimes Q)\propto\r_w(p,2), 2\leq d.
\end{eqnarray}
From (\ref{eq:distillable}), if $\r_w(p_1,2)\otimes\r_w(p_2,2)$ is a tripartite GE state,  $\r_w(p_1,d_1)\otimes\r_w(p_2,d_2)$ is also a tripartite GE state.    According to the results of Theorem \ref{eq:theorem5}, we have  $\r_w(p_1,d_1)\otimes\r_w(p_2,d_2), 2\leq d_1,d_2$ is also a tripartite GE state for the region of $p_1, p_2$ in Fig. \ref{fig:p11}.  Because NPT states can be convert into NPT Werner states by LOCC, we have obtained a large set of distillable NPT states whose tensor product is a tripartite GE state.

In the following we discuss the realization of Theorem \ref{eq:theorem5} experimentally.  The realization of two-qubit Werner states has been extensively studied in experiment over the past several years, using techniques of photon polarization, spontaneous parametric down conversion (SPDC), and semiconductor quantum dot  \cite{BarbieriGeneration,KumanoNonlocal,CinelliParametric,Zhang2002Experimental}. So it is feasible to practically detect the genuine entanglement of the state $\r_{ABC}=\a_{AC_1}\otimes\b_{BC_2}$ with the two bipartite NPT entangled states $\a$ and $\b$. First of all in theory, we respectively convert them into two Werner states under twirling operations $U\otimes U$ \cite{WernerQuantum}. This is a two-qubit local unitary operation and may be implemented effectively. Next we experimentally prepare the two Werner states in system $AC_1$ and $BC_2$, respectively, using the techniques in \cite{BarbieriGeneration,KumanoNonlocal,CinelliParametric,Zhang2002Experimental}. Note that their parameters $p_1,p_2$ in \eqref{eq:rwpd} may be influenced due to white noise. By measuring the Werner states using quantum tomography, we may determine whether the parameters $p_1,p_2$ are restricted in the interval in Theorem \ref{eq:theorem5} (ii). If this is true then we may detect their genuine entanglement using Theorem \ref{eq:theorem5} (ii). It will support Conjecture \ref{conj:1} from a practical point of view.

So far we have investigated Conjecture \ref{conj:1} for NPT entangled states $\a$ and $\b$. Finally we discuss Conjecture \ref{conj:1} when one of $\a$ and $\b$ is a PPT entangled state. It follows from Theorem \ref{eq:theorem range} that Conjecture \ref{conj:1} holds when one of $\a,\b$ satisfies that its range is not spanned by product vectors. For example, the PPT entangled states constructed from the unextendible product bases (UPBs) are such states \cite{DivincenzoUnextendible}. Another example is the so-called completely symmetric state $\a=\sum_{i=0}^6\lambda_{i}\ket{\phi_i}\bra{\phi_i}-\l\proj{\phi_7}$ \cite{ChenSeparability}, where $\l_i,\l>0$, $\ket{\phi_i} = \ket{x_i,x_i},\;\text {for } i=0,1,\ldots, 6,$
and
\begin{equation}
\label{multiSsepv5:eq:3}
\begin{split}
\ket{x_i}& = \ket{i},i=0,1,2,3\\
\ket{x_4}&= \ket{0}+\ket{1}+\ket{2}+\ket{3},\\
\ket{x_5}&= \ket{0}+2\ket{1}+3\ket{2}+4\ket{3},\\
\ket{x_6}&= \ket{0}-2\ket{1}+3\ket{2}-4\ket{3}.\\
\end{split}
\end{equation}
It has been proven that \cite{ChenSeparability} by choosing a positive constant $\l$ we obtain that $\a$ is a $4\times4$ PPT entangled state of rank six, and the range of $\a$ has no product vectors. Since the coefficients $\l_i$'s of $\s$ are arbitrary positive numbers, we have constructed a family of PPT entangled states $\a$ satisfying Conjecture \ref{conj:1}.

The third example is the state
\begin{eqnarray}
\b
\label{eq:sn2-1}
&=&(\ket{00}+\ket{11}+\ket{22})
(\bra{00}+\bra{11}+\bra{22})	
\notag\\
&+&
(\ket{01}+\ket{10}+\ket{33})
(\bra{01}+\bra{10}+\bra{33})	
\notag\\
&+&
\proj{12}+\proj{13}+
\proj{30}+\proj{21}
\notag\\
&+&
\proj{02}+\proj{20}+\proj{03}+\proj{31}.
\end{eqnarray}
It has been used for the separability criteria using symmetric extension \cite{DohertyComplete}, and the well-known PPT square conjecture recently \cite{Chen2019Positive}.
One can show that $\b$ is a PPT state of rank ten, and the range of $\b$ has only eight linearly independent product vectors. So $\b$ is entangled and violates Theorem \ref{eq:theorem range}. We have shown that such $\b$ and any state $\a$ satisfy Conjecture \ref{conj:1}.
Nevertheless, studying Conjecture \ref{conj:1} for the PPT entangled state $\b$ whose range is spanned by product vectors remains an open problem.

In the next section, we show a different method to detect the lower bound of the GE of quantum state, then we compare the region of detecting GE of state in the Example \ref{eq:example2} and Example \ref{eq:example3}.

\section{The lower bound of detecting GE of quantum state}
\label{sec:4}
In this section we investigate the genuine entanglement of tripartite states in terms of the entanglement measure. As the computation of any proper entanglement measure is in general an NP-hard problem, it is crucial for the quantification of entanglement that reliable lower bounds can be derived. In this section, we present the lower bound of detecting GE of quantum state.

The GE concurrence is proved a well-defined measure \cite{MaMeasure}. For example, the GE concurrence is defined by
$
C_{GE}(\ket{\psi})=\sqrt{\min\{1-\tr\r_1^2,1-\tr\r_2^2,1-\tr\r_3^2\}}
$ 
 for a pure state $\ket{\psi}\in\mathcal{B}(\mathcal{H}_1^d\otimes \mathcal{H}_2^d\otimes \mathcal{H}_3^d)$,
where $\r_i$ is the reduced matrix for the $i$-th subsystem. And for mixed state $\r\in\mathcal{B}(\mathcal{H}_1^d\otimes \mathcal{H}_2^d\otimes \mathcal{H}_3^d)$, the GE concurrence is  
$C_{GE}(\r)=\min\sum_{p_\a,\ket{\psi_\a}}p_\a C_{GE}(\ket{\psi_\a})$.
The minimum is taken over all pure ensemble decompositions of $\r$. We show that the GE concurrence satisfies the following fact.

\begin{theorem}
	\label{thm:theorem3}
	For a tripartite qudit state $\r$, the GE concurrence satisfies the following inequality,
	\begin{eqnarray}
	C_{GE}(\r)\geq\max\{\frac{1}{2\sqrt{2}}\|T_\a^{(123)}\|-\frac{(d-1)}{d}\sqrt{\frac{d+1}{d}},0\}.
	\notag\\
	\end{eqnarray}
\end{theorem}
\begin{proof}
	First, we consider pure states $\r=\ketbra{\psi}{\psi}$. We have $\tr\r^2=1$ and hence
	\begin{eqnarray}
	\label{eq:pure}
	&&
	\frac{1}{d^3}+\frac{1}{2d^2}(\sum(t_i^1)^2+\sum(t_j^2)^2+\sum(t_k^3)^2)\nonumber\\
	&+&
	\frac{1}{4d}(\sum(t_{ij}^{12})^2+\sum(t_{ik}^{13})^2+\sum(t_{jk}^{23})^2)\nonumber\\
	&+&
	\frac{1}{8}\sum(t_{ijk}^{123})^2=1.
	\end{eqnarray}
	We denote $\r_{jk}$ as the reduced density matrix for the subsystems $j\neq k=1,2,3$. Using (\ref{eq:defr}), we have
	\begin{eqnarray}
	\tr\r_1^2=\frac{1}{d}+\frac{1}{2}\|T^{(1)}\|^2,
	\end{eqnarray}
	\begin{eqnarray}
	\tr\r_{23}^2=\frac{1}{d^2}+\frac{1}{2d}(\|T^{(2)}\|^2+\|T^{(3)}\|^2)+\frac{1}{4}\|T^{(23)}\|^2,
	\end{eqnarray}
	\begin{eqnarray}
	\tr\r_2^2=\frac{1}{d}+\frac{1}{2}\|T^{(2)}\|^2,
	\end{eqnarray}
	\begin{eqnarray}
	\tr\r_{13}^2=\frac{1}{d^2}+\frac{1}{2d}(\|T^{(1)}\|^2+\|T^{(3)}\|^2)+\frac{1}{4}\|T^{(13)}\|^2,
	\end{eqnarray}
	and
	\begin{eqnarray}
	\tr\r_3^2=\frac{1}{d}+\frac{1}{2}\|T^{(3)}\|^2,
	\end{eqnarray}
	\begin{eqnarray}
	\tr\r_{12}^2=\frac{1}{d^2}+\frac{1}{2d}(\|T^{(1)}\|^2+\|T^{(2)}\|^2)+\frac{1}{4}\|T^{(12)}\|^2.
	\end{eqnarray}
	Setting $C=\|T^{(123)}\|^2$ in (\ref{eq:pure}), we obtain
	\begin{eqnarray}
	\label{eq:rr1}
	\tr\r^2-\tr\r_1^2
	&=&
	(\frac{3}{d^2}-\frac{2}{d^3}-\frac{1}{d})+\frac{d-1-d^2}{2d^2}\|T^{(1)}\|^2\nonumber\\
	&+&\frac{1}{2}(\frac{1}{d}-\frac{1}{d^2})(\|T^{(2)}\|^2+\|T^{(3)}\|^2)+\frac{1}{8}C.
	\notag\\
	\end{eqnarray}
	If we assume that $\|T^{(2)}\|^2+\|T^{(3)}\|^2=0$, and we use $\|T^{(1)}\|^2\leq\frac{2(d-1)}{d}$ of \cite{Vicente2011Multipartite}, then we obtain that $
	\tr\r^2-\tr\r_1^2
	\geq
	\frac{-(1+d)(d-1)^2}{d^3}+\frac{1}{8}C,
	$
	and the same lower bound of $\tr\r^2-\tr\r_2^2$ and $\tr\r^2-\tr\r_3^2$. Thus
	\begin{eqnarray}
	C^2_{GE}(\ket{\psi})
	&=&
	\min\{1-\tr\r_1^2,1-\tr\r_2^2,1-\tr\r_3^2\}\nonumber
	\\
	&\geq&
	\max\{\frac{-(1+d)(d-1)^2}{d^3}+\frac{1}{8}C,0\}.
	\end{eqnarray}
	Next we consider the mixed quantum state $\r\in \mathcal{B}(\mathcal{H}_1^d\otimes \mathcal{H}_2^d\otimes \mathcal{H}_3^d)$. Let $\r=\sum p_\a\ketbra{\psi_\a}{\psi_\a}$ be the optimal ensemble decomposition of $\r$. We obtain that
	\begin{widetext}
	\begin{eqnarray}
	C_{GE}(\r)
	&=&
	\sum p_\a C_{GE}(\ket{\psi_\a})\nonumber
	\\
	&\geq&
	\max\{\frac{1}{2\sqrt{2}}\|T_\a^{(123)}\|-\frac{(d-1)}{d}\sqrt{\frac{d+1}{d}},0\}\times\sum p_\a\nonumber
	\\
	&=&
	\max\{\frac{1}{2\sqrt{2}}\|T^{(123)}\|-\frac{(d-1)}{d}\sqrt{\frac{d+1}{d}},0\},
	\end{eqnarray}
	\end{widetext}
	where we have used the convexity of Frobenius norm and $\sqrt{x-y}\geq\sqrt{x}-\sqrt{y}$ for $x>y>0$.
\end{proof}

In the following we present two examples showing the power of Theorem \ref{thm:theorem3}.
\begin{example}
	\label{eq:example2}
	Consider the mixture of the GHZ state and W state in three-qubit quantum systems $\r=\frac{1-x-y}{8}I+x\ketbra{GHZ}{GHZ}+y\ketbra{W}{W}$, where $\ket{GHZ}=\frac{1}{\sqrt{2}}(\ket{000}+\ket{111})$ and $\ket{W}=\frac{1}{\sqrt{3}}(\ket{001}+\ket{010}+\ket{100})$.
\end{example}

To explain the example, we choose
$\l_1=
\bma
1&0\\0&1
\ema,
\l_2=
\bma
0&-i\\
i&0
\ema
$ and $
\l_3=
\bma
0&1\\
1&0
\ema
$
as generators $\l$-matrices of $SU(2)$. According to Theorem \ref{thm:theorem3},
then we have
\begin{eqnarray}
\label{eq:46}
\mathcal{C}_{GE}(\r)
&=&
\sum p_\a\mathcal{C}_{GE}(\ket{\psi_{\a}})\nonumber\\
&\geq&
\max\{\frac{1}{2\sqrt{2}}\|T^{(123)}(\r)\|-\frac{1}{2}\sqrt{\frac{3}{2}},0\}\nonumber
\\
&=&
\max\{\frac{1}{2\sqrt{6}}(\sqrt{12x^2+11y^2}-3),0\}.
\end{eqnarray}

From (\ref{eq:46}),  we obtain that Theorem \ref{thm:theorem3} can detect GE of state $\r$ when $\frac{1}{2\sqrt{6}}(\sqrt{12x^2+11y^2}-3)\geq0$ in Fig. \ref{fig:p2}. So we obtain that the lower bound of detecting GE of state $\r$. (see Fig. \ref{fig:p3}).
\begin{figure}[t]
	\includegraphics[scale=0.5,angle=0]{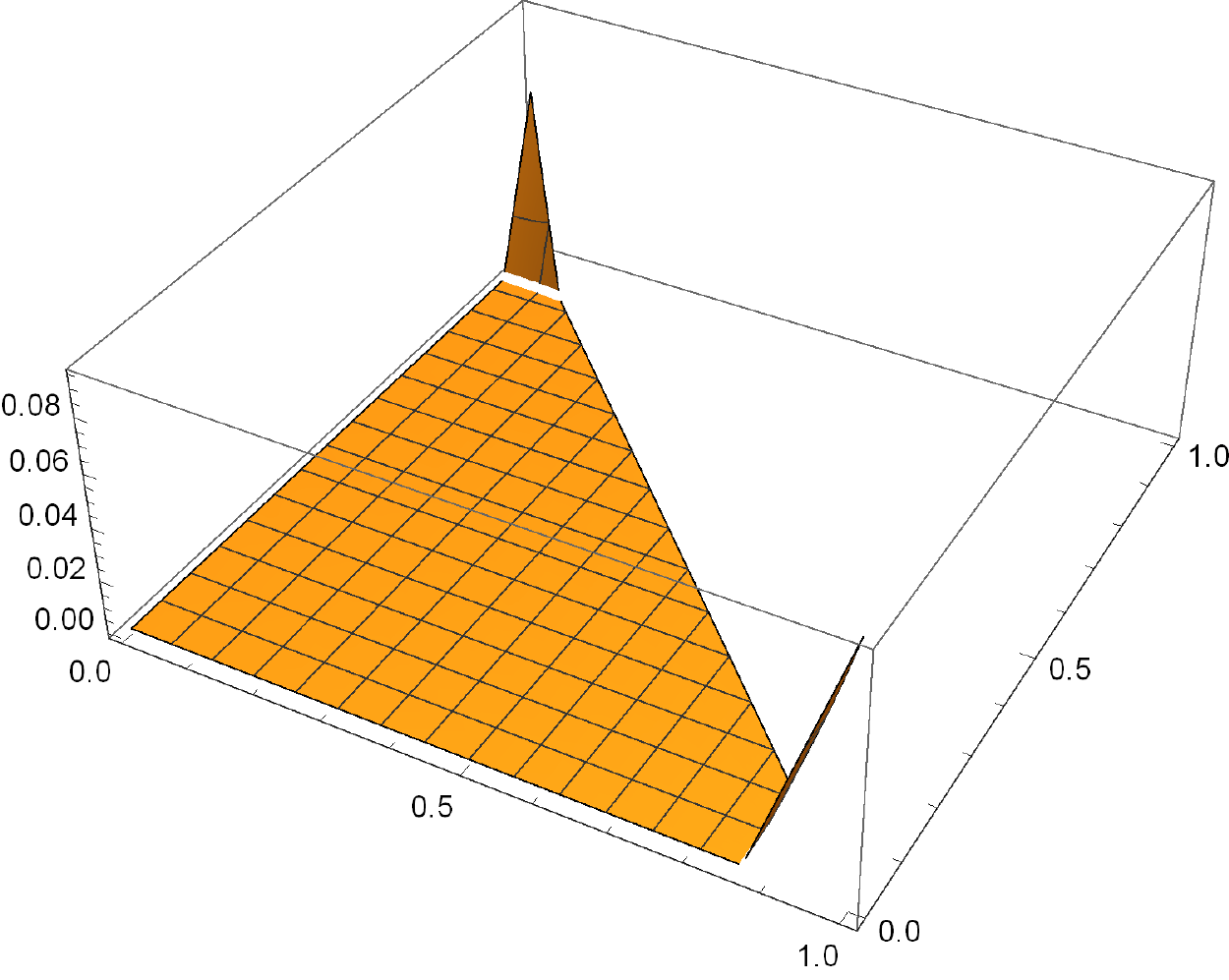}
	\caption{Lower bound of GE concurrence in Theorem \ref{thm:theorem3}} of state $\r$ in Example \ref{eq:example2}.
	\label{fig:p2}
\end{figure}
\begin{figure}[t]
	\includegraphics[scale=0.5,angle=0]{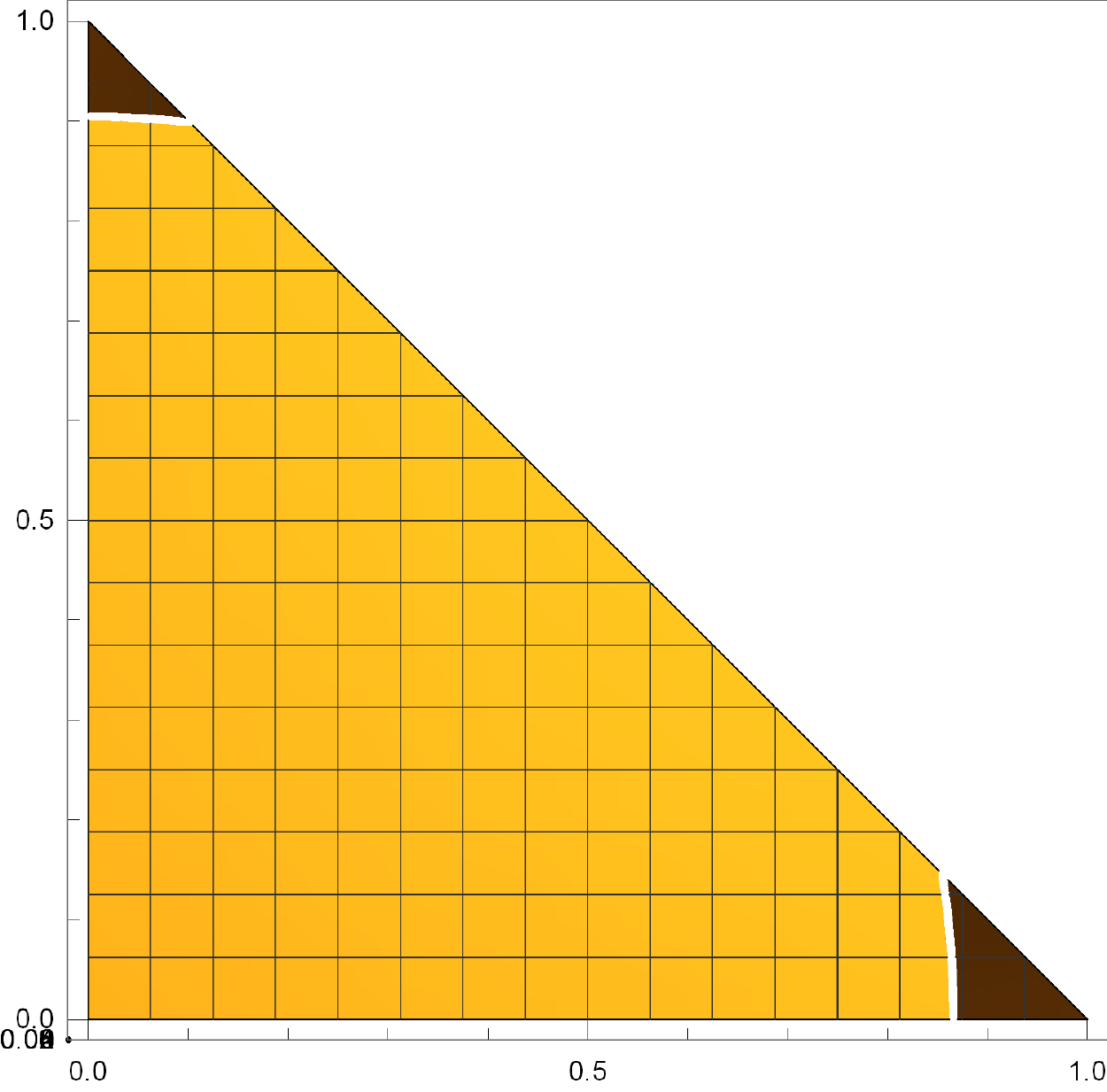}
	\caption{The lower bound for GE concurrence is equal to zero, then GE can be detected For $0.854794<x\leq1$ and $0.898272<y\leq1$ in Theorem \ref{thm:theorem3}.}
	\label{fig:p3}
\end{figure}

\begin{example}
	\label{eq:example3}
	Consider a mixed state in three-qutrit quantum systems $\r=\frac{1-x}{27}I+x\ketbra{\psi}{\psi}$, where $\ket{\psi}=\frac{1}{\sqrt{3}}(\ket{012}+\ket{021}+\ket{111})$.
\end{example}

To detect the genuine entanglement of $\r$, then we respectively use the ways in Theorem \ref{eq:theorem2}, Theorem \ref{eq:theorem1} and Theorem \ref{thm:theorem3}, the results are as follows,

(i) By using Theorem \ref{eq:theorem2}, we have
\begin{eqnarray}
\label{eq:equation1}
&&
M_k(\r)-\frac{8}{9}\sqrt{\frac{2}{3}}(1+4\sqrt{2})\nonumber
\\&&
\hspace{-0.4cm}=
\frac{1}{3}(\|T_{\underline{1}23}\|_8+\|T_{\underline{2}13}\|_8+\|T_{\underline{3}12}\|_8)-\frac{8}{9}\sqrt{\frac{2}{3}}(1+4\sqrt{2})\nonumber
\\&&
\hspace{-0.4cm}\approx
4.30179x-3.628874>0.
\end{eqnarray}

(ii) By using Theorem \ref{eq:theorem1}, we have
\begin{eqnarray}
\|T_{i_1i_2i_3}\|-\frac{8}{3}\sqrt{\frac{2}{3}}
&\approx&
2.17732(x-1)>0.
\end{eqnarray}

(iii) By using Theorem \ref{thm:theorem3}, we have
\begin{eqnarray}
C_{GE}(\r)
&\geq&
\max\{\frac{1}{2\sqrt2}\|T^{(123)}\|-\frac{2}{3}\sqrt{\frac{4}{3}},0\}\nonumber
\\
&\approx&
\max\{-0.7698 + 0.7698 x,0\}.
\end{eqnarray}

In Fig. \ref{fig:p4}, the lower bound of GE concurrence in Theorem \ref{eq:theorem2} can detect GE better than Theorem \ref{eq:theorem1} and the Theorem \ref{thm:theorem3}.

\begin{figure}[t]
	\includegraphics[scale=0.5,angle=0]{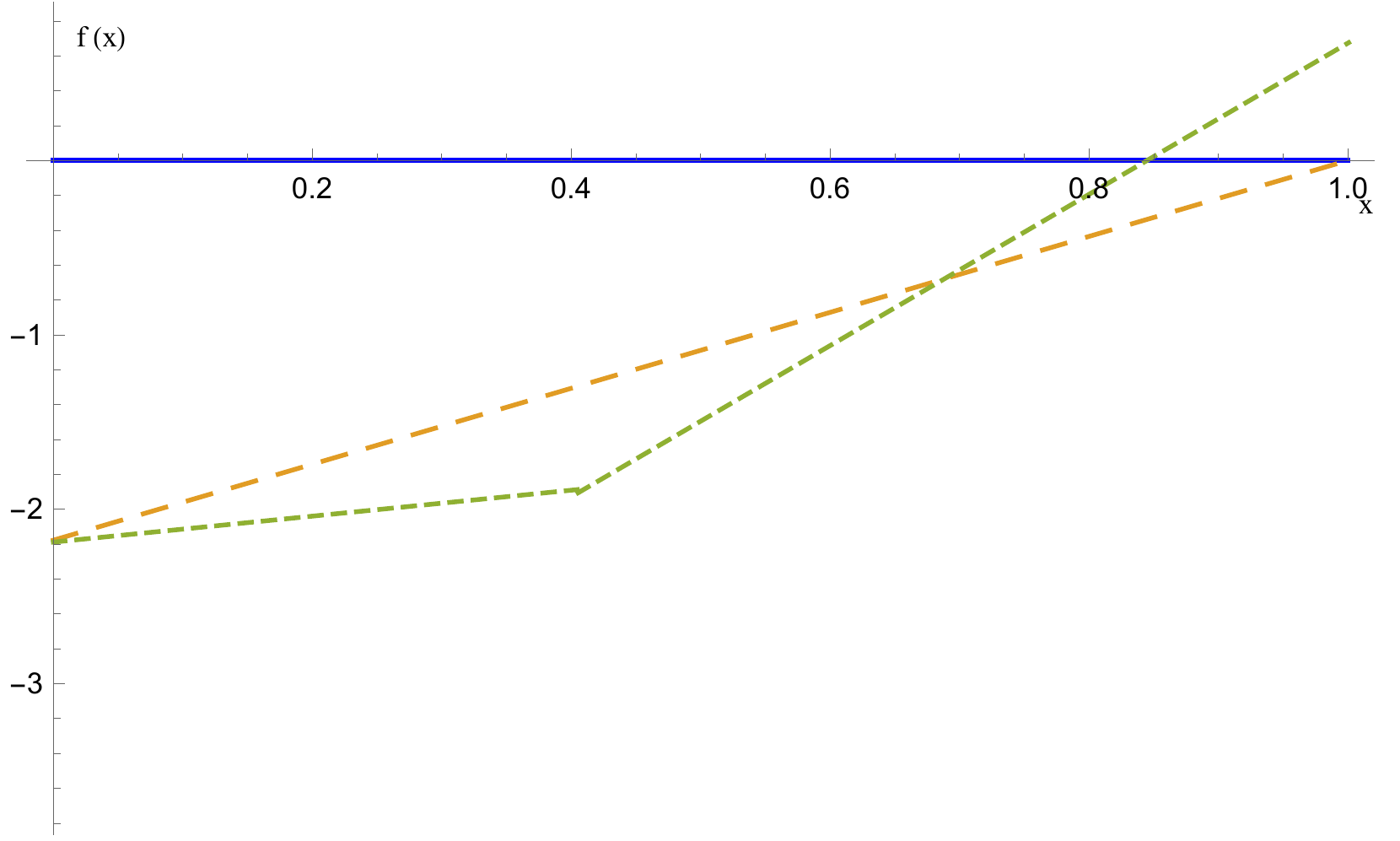}
	\caption{GE detection: dashed brown line by Theorem \ref{thm:theorem3} and solid blue line in Theorem \ref{eq:theorem1} cannot detect GE for the whole region of $x$. When $d=3, k=4$, Theorem \ref{eq:theorem2} in this paper detect GE for $0.843573<x\leq1$.}
	\label{fig:p4}
\end{figure}

To conclude, we point out that some results of this section have corrected some suspicious claims from \cite{Ming2017Measure}. First in Theorem \ref{thm:theorem3},
our article corrects the maximum lower bound for genuine multipartite entanglement concurrence in terms of the norms of the correlation tensors of $T_\a^{(123)}$. Next, in  Example \ref{eq:example2}, we have corrected the lower bound for genuine multipartite entanglement concurrence. In Example \ref{eq:example3}, we have corrected the representation of $M(\r)$ and the lower bound for GE concurrence. Then we obtain that the result of Theorem \ref{thm:theorem3} cannot detect tripartite GE in  Example \ref{eq:example3}.

\section{Conclusions}
\label{sec:con}
	
In this paper, we have asked whether the product of two entangled states $\a_{AC_1}$, $\b_{BC_2}$ is still tripartite genuine entangled, which has been formulated by Conjecture \ref{conj:1}. Then we have discussed the realization of Conjecture \ref{conj:1} in experiment by using two Werner states. By restricting the interval of parameter $p_1$ and $p_2$, we can detect the tripartite genuine entanglement.
Besides, we also have corrected the lower bound for genuine multipartite entanglement concurrence of any quantum states, and have detected genuine entanglement in two examples by using it.

A direct open problem from this paper is to keep studying Cojecture \ref{conj:1} for region of parameter $p$ of two bipartite NPT states and more general cases. However, it is also very interesting to find out a counterexample, because it shows the physical difference between bipartite and tripartite genuine entanglement.

\section*{Acknowledgments}
	\label{sec:ack}	
	
Authors were supported by the  NNSF of China (Grant No. 11871089), and the Fundamental Research Funds for the Central Universities (Grant Nos. KG12080401 and ZG216S1902).

\appendix


	
	
\section{GE detection of $\r$ by Other ways compare with Theorem \ref{eq:theorem5}}
\label{eq:appendix2}	
To compare the region of $x$ for detecting GE of state $\r$ in Theorem \ref{eq:theorem5}, there are two different ways (i), (ii) as follows,

(i) By using  Theorem \ref{eq:theorem1}, we have
\begin{eqnarray}
\label{eq:eg101}
\|T_{i_1i_2i_3}\|-\frac{3}{2}\sqrt{\frac{5}{2}}
&=&
\sqrt{3}\frac{\sqrt{p_2^2+p_1p_2^2+p_1^2(1+p_2+2p_2^2)}}{(2+p_1)(2+p_2)}x-\frac{3}{2}\sqrt{\frac{5}{2}}>0
\notag\\
\end{eqnarray}
where $p_1,p_2\in[-1,-\frac{1}{2})$. To maximum detect the GE, then we find the maximum slope of $x$ in (\ref{eq:eg101}). By analysing the maximum of the coefficient of $x$, we have the maximum of $\sqrt{3}\frac{\sqrt{p_2^2+p_1p_2^2+p_1^2(1+p_2+2p_2^2)}}{(2+p_1)(2+p_2)}x$ is $\sqrt{6}$, when $p_1=p_2=-1$. then Theorem \ref{eq:theorem1} can detect the GE of state $\r$  for $0.968246<x\leq1$.

\begin{figure}[t]
	\includegraphics[scale=0.6,angle=0]{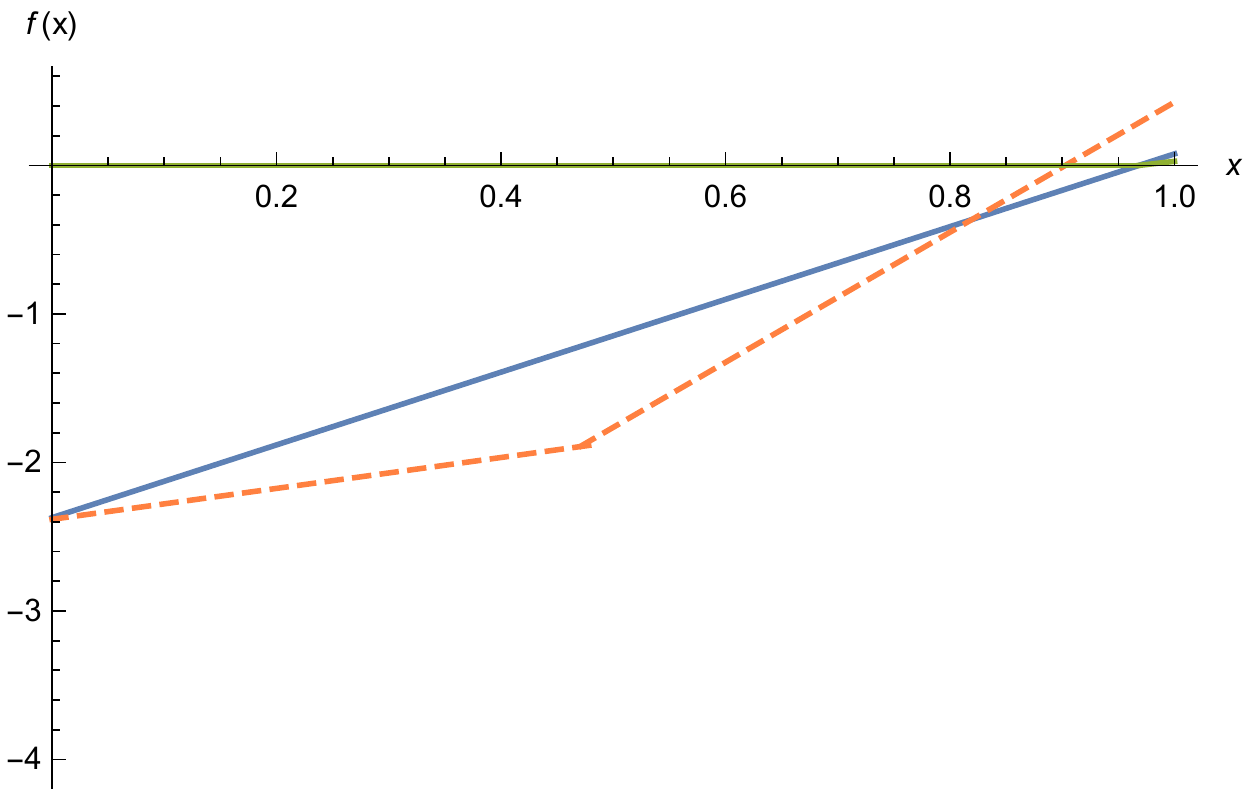}
	\caption{Under the condition that $p_1=p_2=-1$, then we have the region of detecting GE of state $\r$ is maximum. One finds in this figure that solid green line by Theorem \ref{thm:theorem3} and solid blue line in Theorem \ref{eq:theorem1} can detect GE of $\r$ in the same region $(0.968246,1]$. The dashed brown line by Theorem \ref{eq:theorem2} can detect GE of $\r$ in $(0.902646,1]$.}
\label{fig:p5}
\end{figure}

On the other hand, in the region of $x$ that GE of state $\r$ can be detected, we consider the boundary condition that can detect GE of state $\r$. From (\ref{eq:eg101}), we obtain that the GE of $\r$ can be detected when
\begin{eqnarray}
\sqrt{3}\frac{\sqrt{p_2^2+p_1p_2^2+p_1^2(1+p_2+2p_2^2)}}{(2+p_1)(2+p_2)}x\geq\frac{3}{2}\sqrt{\frac{5}{2}}.
\end{eqnarray}

Then we have the boundary condition that can detect GE of state $\r$ is
\begin{eqnarray}
\label{eq:boundary1}
\sqrt{3}\frac{\sqrt{p_2^2+p_1p_2^2+p_1^2(1+p_2+2p_2^2)}}{(2+p_1)(2+p_2)}x=\frac{3}{2}\sqrt{\frac{5}{2}}.
\end{eqnarray}
From (\ref{eq:boundary1}), then                                                                                                                                                                                                                                                                                                                                                                                                                                                                                                                                                                                                                                                                                                                                                                                                                                                                                                                                                                                                                                                                                                                                                                                                                                                                                                                                                                                                                                                                                                                                                                                                                                                                                                                                                                                                                                                                                                                                                                                                                                                                                                                                                                                                                                                                                                                                                                                                                                                                                                                                                                                                                                                                                                                                                                                                                                                                                                                                                                                                                                                                                                                                                                      we obtain that negative correlation between $p_1$ and $p_2$ for $-1\leq p_1,p_2\leq-0.981475$. In the region of $p_1, p_2$, the GE of state $\r$ can be detected. Specially, when $p_1=p_2=-1$, the region of GE detection of state $\r$ is maximum.




(ii) By using Theorem \ref{thm:theorem3}, we have
\begin{eqnarray}
\|T^{(123)}\|=\sqrt{3}\frac{\sqrt{p_2^2+p_1p_2^2+p_1^2(1+p_2+2p_2^2)}}{(2+p_1)(2+p_2)}x.
\end{eqnarray}
Then we obtain
\begin{eqnarray}
\label{eq:GE1}
C_{GE}(\r)
&\geq&
\max\{\frac{1}{2\sqrt{2}}\|T^{(123)}\|-\frac{3}{4}\sqrt{\frac{5}{4}},0\}\nonumber
\\
&=&
\max\{\frac{\sqrt{3}}{2\sqrt{2}}\frac{\sqrt{p_2^2+p_1p_2^2+p_1^2(1+p_2+2p_2^2)}}{(2+p_1)(2+p_2)}x-\frac{3}{4}\sqrt{\frac{5}{4}},0\}.
\notag\\
\end{eqnarray}
From the maximum of $\sqrt{3}\frac{\sqrt{p_2^2+p_1p_2^2+p_1^2(1+p_2+2p_2^2)}}{(2+p_1)(2+p_2)}x=\sqrt{6}$ when $p_1=p_2=-1$, then we have $\max\|T^{(123)}\|=\sqrt{6}x$. So we obtain that Theorem \ref{thm:theorem3} can detect the GE of $\r$ for $0.968246<x\leq1$. 

Thus, by comparing the region of $x$ for detecting GE of state $\r$ in Fig. \ref{fig:p5}, we obtain that the Theorem \ref{eq:theorem5} can detect GE of state $\r$ for the maximum region of $x$.

\bibliographystyle{unsrt}

\bibliography{ge_error_of_liming}

\end{document}